\documentclass[a4paper]{article}
\usepackage[english]{babel}
\usepackage{float}
\usepackage{amssymb}
\usepackage{amsmath}
\usepackage{url}          
\usepackage{hyperref}
\usepackage{fullpage}
\usepackage{enumitem}
\usepackage{amsthm}
\usepackage{subcaption}
\usepackage{comment}
\usepackage[ruled]{algorithm2e}
\usepackage{float}
\usepackage{multicol}
\usepackage{color}
\usepackage{graphicx}
\newcounter{multifig}

\newcommand{\figcaption}[1]% #1 = text
{\stepcounter{multifig}
	\addcontentsline{lof}{figure}{\string\numberline {\arabic{multifig}}{\ignorespaces #1}}
	Figure \arabic{multifig}: #1}

\newtheorem{theorem}{Theorem}[section]

\newtheorem{remark}[theorem]{Remark}
\newtheorem{example}[theorem]{Example}

\title{Distributed Multi-resource Allocation with Little Communication Overhead}

\author{Syed Eqbal Alam\thanks{Concordia Institute for Information Systems Engineering,
		Concordia University, 
		Montreal, Quebec, Canada, email:  sy\_al@encs.concordia.ca, jiayuan.yu@concordia.ca}, 
	Robert Shorten\thanks{School of Electrical, Electronic
		and Communications Engineering, University College Dublin, Dublin, Ireland, email: robert.shorten@ucd.ie} ,
	Fabian Wirth\thanks{Faculty of Computer Science and Mathematics, University
		of Passau, Passau, Germany, email:  fabian.wirth@uni-passau.de} , and
	Jia Yuan Yu$^*$ }
\date{}

\begin{document}
	\maketitle
	
	\begin{abstract}  % put your abstract here!
		We propose a distributed algorithm to solve a special
		distributed multi-resource allocation problem with no direct
		inter-agent communication. We do so by extending a recently introduced
		additive-increase multiplicative-decrease (AIMD) algorithm, which only
		uses very little communication between the system and agents. Namely,
		a control unit broadcasts a one-bit signal to agents whenever one of
		the allocated resources exceeds capacity. Agents then respond to this
		signal in a probabilistic manner. In the proposed algorithm, each
		agent is unaware of the resource allocation of other agents. We also
		propose a version of the AIMD algorithm for multiple {\em binary}
		resources (e.g., parking spaces). Binary resources are indivisible
		unit-demand resources, and each agent either allocated one unit of the
		resource or none.  In empirical results, we observe that in both
		cases, the average allocations converge over time to optimal allocations.
	\end{abstract}
	
	\maketitle
	
	\section{Introduction}
	Distributed optimization has numerous applications in many different
	areas. These include: sensor networks; Internet of Things; smart grid;
	and smart transportation. In many instances, networks of agents
	achieve optimal allocation of resources through regular communication
	with each other and/or with a control unit. Details of such
	distributed optimization problems can be found in (among others)
	~\cite{Lin2014}, \cite{Deilami2011}, \cite{Wirth2014}, \cite{Johansson2008}, \cite{Blondel2005}, \cite{Chang2014}, \cite{Duchi2012}, \cite{Shi2013}, \cite{Wang2010} and the papers cited therein.
	
	In some applications, groups of coupled multiple resources, must be
	allocated among competing agents. Generally speaking, such problems
	are more difficult to solve in a distributed manner than those with a
	single resource. This is particularly true when communication between
	agents is constrained---either through limitations of communication
	infrastructure, or due to privacy considerations. Motivated by such
	applications, we wish to find an algorithm that is tailored for such
	scenarios and where there is no requirement of inter-agent
	communication. Our starting point is \cite{Wirth2014},
	\cite{Studli2015_2}. The authors of these papers demonstrate that simple
	algorithms from Internet congestion control can be used to solve
	certain optimization problems. Our contribution here is to demonstrate
	that the ideas therein extend to a much broader (and more useful)
	class of optimization problems.
	
	Roughly speaking, in \cite{Wirth2014}, the iterative distributed optimization
	algorithm works as follows. Agents continuously acquire an increasing
	share of the shared resource. When the aggregate agent demand exceeds
	the total capacity of resources, then the control unit sends a one bit
	capacity event notification to all competing agents and the agents
	respond in a probabilistic manner to reduce demand. By judiciously
	selecting the probabilistic manner in which agents respond, a portfolio of
	optimization problems can be solved in a stochastic and distributed
	manner. Our proposed algorithm extends this previous work. It builds
	on the choice of probabilistic response strategies described therein
	but is different in the sense that we generalize the approach to deal
	with {\em multiple resource constraints}. The communication overhead is quite low in our proposed
	solutions and the communication complexities are independent of the
	number of agents competing for resources.
	
	To be precise, suppose $n$ agents compete for $m$ divisible resources
	$R^1, R^2, \ldots, R^m$ with capacity $C^1, C^2, \ldots, C^m$,
	respectively. We use $i \in \{1, 2,\ldots, n\}$ as an index for agents and
	$j \in \{1, 2,\ldots, m\}$ to index the resources. Each agent has a cost
	function $f_i: \mathbb{R}^m \to \mathbb R$ which associates a cost to a
	certain allotment of resources and which may depend on the agent. We
	assume that $f_i$ is twice continuously differentiable, convex, and
	increasing in all variables, for all $i$. For all $i$ and $j$, we denote by $x_i^j \in \mathbb{R}_+$ the
	amount of resource $R^j$ allocated to agent $i$.  We are interested in the
	following optimization problem of {\em multi-resource}
	allocation:
	\begin{align} \label{obj_fn1}
	\begin{split}
	\min_{{x}^1_1, \ldots, {x}^m_n} \quad &\sum_{i=1}^{n} f_i(x^1_i, x^2_i,
	\ldots, x^m_i),    		\\
	\mbox{subject to} \quad
	&\sum_{i=1}^{n} x^j_i = C^j, \quad j \in \{1,  \ldots, m\},		\\
	&x^j_i \geq 0, \quad i \in \{1, \ldots, n\}, \ j \in \{1, \ldots, m\}.
	\end{split}
	\end{align}
	Note that there are $nm$ decision variables $x^j_i$ in this optimization problem.
	We denote the solution to the minimization problem by
	$x^{*} \in \mathbb{R}_+^{nm}$, where $x^* = (x_1^{*1}, \ldots, x_n^{*m})$. By compactness of the constraint set optimal
	solutions exist. We also assume strict convexity of the cost
	function $\sum_{i=1}^{n} f_i$, so that the optimal solution is unique.
	
	We propose iterative schemes that achieve optimality for the
	long-term average of allocations. Suppose $\mathbb{N}$ denotes the set of natural numbers and $k \in
	\mathbb{N}$ denotes the time steps. To this end we will denote by $x_i^j(k)$
	the amount of resource allocated at the (discrete) time step $k$. The
	average allocation is (for $i \in \{ 1, \ldots, n \}, \ j \in \{ 1, \ldots,m \}$, and $k \in \mathbb{N}$)
	\begin{align} \label{average_eqn}
	\overline{x}^j_i(k)=\frac{1}{k+1} \sum_{\ell=0}^k x^j_i(\ell).
	\end{align}
	
	In the sequel two main problems are addressed. A common feature of the two
	is that the agents do not communicate with each other. The only
	information transmitted by a control unit is the occurrence of a
	capacity event. The problems are
	
	\begin{itemize}
		\item[I.] \textbf{(Divisible Resources)} In this setting, mentioned in
		Section~\ref{divisible_mul_res}, it is assumed that resources are divisible and
		agents can obtain any amount in $[0,\infty)$. The problem is to derive
		an iterative scheme, such that
		\begin{align}
		\label{eq:longtermopt}
		\lim_{k\to\infty} \overline{x}(k) = x^{*}.
		\end{align}
		
		\item[II.] \textbf{(Unit-Demand Resources)} In this setting, treated in
		Section~\ref{bin_imp}, only $0$ or $1$ unit of the resource can be
		allocated at any given time step. In the long-term average, which is
		defined as in \eqref{average_eqn} we may still
		achieve a non-unit optimal point, but the allocation at each
		particular time step becomes more involved. The aim is still to achieve the
		optimal point on long-term average, as defined in \eqref{eq:longtermopt}.
	\end{itemize}
	In both the settings, we use the consensus of the derivatives of the cost functions of all agents competing for a particular resource to show the optimality. Suppose $x_i^* = (x_i^{*1}, \ldots, x_i^{*m}) \in \mathbb{R}_+^m$ denotes the optimal values of agent $i$ for all $m$ resources in the system. We say that the derivatives with respect to resource $R^j$ are in consensus if
	\begin{align}\label{optimality}
	\nabla_j f_i(x_i^{*}) = \nabla_j f_u(x_u^{*}), \mbox{ for } i,u \in \{1,2, \ldots, n\} \mbox{ and } j \in \{1,2,\ldots, m\}.
	\end{align}
	
	As a brief background about allocation of indivisible unit-demand
	resource; it is an active area of research, which goes back to the work
	of Koopmans and Beckmann \cite{Tjalling1957}. To get
	the details of recent works on allocation techniques of indivisible
	unit-demand resources, the interested readers can look into
	\cite{Aziz2014} and the papers cited therein. 
	The allocation of unit-demand resources of Section~\ref{bin_imp} is a generalization of
	\cite{Griggs2016} but different in the sense that the cost functions
	and the constraints used therein depend on allocation of multiple
	indivisible unit-demand resources and proposed for general
	application settings. The proposed algorithm works as follows;
	the control unit broadcasts a normalization signal to all the agents
	in system time to time, control unit updates this normalization
	signal using the utilization of resources at earlier time
	step. After receiving this signal the algorithm of each agent
	responds in a probabilistic manner either to demand for the resource
	or not. This process repeats over time.

	\section{Motivation} \label{motivation}
	As in \cite{Wirth2014} we use a modified {\em additive-increase
		multiplicative-decrease} (AIMD) algorithm. By way of background,
	the AIMD algorithm was proposed in the context of congestion avoidance
	in transmission control protocol (TCP) \cite{Chiu1989}. It involves
	two phases; the first is {\em additive-increase} phase, and second is
	{\em multiplicative-decrease} phase. In additive-increase phase, a
	congestion window size increases (resource allocation) linearly until
	an agent is informed that there is no more resource available. We call
	this a {\em capacity event}. Upon notification of a capacity event,
	the window size is reduced abruptly. This is called the {\em
		multiplicative-decrease} phase. Motivated by this basic algorithm,
	we propose a modified algorithm for solving a class of optimization
	problems. Here, as in AIMD, agents keep demanding the resources of
	different types until a capacity event notification is sent by
	the control unit to them. However, after receiving a capacity event
	notification the agents toss a coin to determine whether to reduce
	their resource demand abruptly or not. In our context, the probability of
	responding is selected in a manner that ensures that our algorithm
	asymptotically solves an optimization problem. The precise details in
	which to select these probabilities is explained in
	Section~\ref{divisible_mul_res}.
	
	The AIMD algorithm has
	been explored and used in many application areas. See
	for example the recent book by Corless et al.~\cite{Corless2016} for
	an overview of some applications; the papers \cite{Wirth2014} for distributed optimization applications;
	\cite{Crisostomi2014} for microgrid applications; and \cite{Cai2005}
	for multimedia applications. The recent literature is also rich with
	algorithms that are designed for distributed control and optimization
	applications. Significant contributions have been made in many
	communities; including, networking, applied mathematics, and control
	engineering. While this body of work is too numerous to enumerate, we point the interested readers to the works of Nedic
	\cite{Nedic2009, Nedic2011}; Cortes \cite{KIA2015}; Jadbabaie and
	Morse\cite{Jadbabaie2003}; Bullo \cite{Bullo2011}; Pappas
	\cite{Pappas2017}, Bersetkas~\cite{Bertsekas2011}; Tsitsiklis
	\cite{Blondel2005} for recent contributions. From a technical
	perspective, much of the recent attention has focussed on distributed
	primal-dual methods and application of {\em alternating direction
		method of multipliers} (ADMM) based techniques. A survey of some of
	this related work is given in \cite{Wirth2014}. Our proposed algorithm is motivated by the fact that we are interested
	in allocating multiple resource types. Such a need arises in several
	areas, some of which are described below.
	
	\begin{example}[Cloud computing]
		In cloud computing, on-demand access is provided to
		distributed facilities like computing resources
		\cite{Armbrust2010}. Resources are shared among the users and each
		user gets a fraction of these resources over time. For example,
		companies may compete for both memory and CPU cycles in such
		applications. Practically, users of cloud services do not interact
		with each other and their demand is only known to themselves (due to
		privacy concerns when different, perhaps competing, companies use the
		same shared resource). In such a setting our proposed algorithm can be
		useful for allocating resources in a way that requires little
		communication with the control unit, and such that there is no
		inter-user communication.
	\end{example}
	
	\begin{example}[Car sharing]
		Consider now a situation where a city sets aside
		a number of free (no monetary cost) parking spaces and charge points
		to service the needs of car sharing clubs in cities. An example of a
		city that implements such a policy is Dublin in Ireland. Now suppose
		that there are a number of clubs competing for such resources via
		contracts. A city must decide which spaces, and which charge points,
		to allocate to each club. Clearly, in such a situation, resources
		should be allocated in a distributed manner that preserves the
		privacy of individual companies, but which also maximizes the
		benefit to a municipality.
	\end{example}
	
	\section{Preliminaries} \label{prelim}
	
	In addition to the notations already introduced, we let
	$\mathcal{N}:=\{1, 2, \ldots, n\}$;
	$\mathcal{M}:=\{1, 2, \ldots, m\}$.
   For a sufficiently smooth function $f:\mathbb R^m \to \mathbb R$, we
	denote by $\nabla_j f$ the $j$th partial derivative of $f$, for $j \in
	\mathcal{M}$ and the Hessian of $f$ is denoted by $\nabla^2 f$.
	
	\subsection{A primer on AIMD}
	The AIMD algorithm is of interest because it can be tuned to achieve
	optimal distribution of a single resource among a group of agents. To this
	end no inter-agent communication is necessary. The agents just receive
	capacity signals from a central unit and respond to it in an stochastic
	manner. This response can be tuned so that the long-term average
	optimality criterion (cf. \eqref{eq:longtermopt}) can be achieved.
	
	In AIMD each agent follows two rules of action at each time step:
	either it increases its share of the resource by adding a fixed amount while total demand is
	less than the available capacity, or it reduces its share in a
	multiplicative manner when notified
	that global capacity has been reached. In the additive increase (AI) phase of the algorithm agents probe the
	available capacity by continually increasing their share of the
	resource.  The multiplicative decrease\index{multiplicative-decrease}
	(MD) phase occurs when agents are notified that the capacity limit has
	been reached; they respond by reducing their shares, thereby freeing
	up the resource for further distribution. This pattern is repeated by every agent as
	long as the agent is competing for the resource. The only information
	given to the agents about availability of the resource is a
	notification when the collective utilization of the resource achieves
	some capacity constraint.  At such times, so called {\em capacity
		events}, some or all agents are instantaneously informed that
	capacity has been reached.  The mathematical description of the basic continuous-time AIMD model
	is as follows.  Assume $n$ agents, and denote the share of the
	collective resource obtained by agent $i$ at time $t\in \mathbb{R}_+$ by
	$x_i(t) \in \mathbb{R}_+$.  Denote by $C$ the total capacity of the resource available to the
	entire system (which need not be known by the agents).  The capacity constraint requires that
	$\sum_{i=1}^{n} x_i(t) \le C$ for all $t$. As all agents are continuously
	increasing their share this capacity constraints will be reached
	eventually. We denote the times at which this happens by $t_k, k\in \mathbb{N}$.
	At time $t_k$ the global utilization of the resource reaches
	capacity, thus
	\begin{align*}
	\sum_{i=1}^n x_i(t_k) &= C.
	\end{align*}
	When capacity is achieved, some agents
	decrease their share of the resource.
	The instantaneous decrease of the share for agent $i$ is
	defined by:
	\begin{align} \label{MD} x_i(t_k^{+}) := \lim_{t \rightarrow t_k,
		\, t > t_k} x_i(t) = \beta_i x_i(t_k),
	\end{align}
	where $\beta_i$ is a constant satisfying $ 0\le \beta_i <1.$
	In the simplest version of the algorithm, agents are assumed to
	increase their shares at a constant rate in the AI phase: 
	\begin{align} \label{AI} x_i(t ) = \beta_i x_i(t_k) + \alpha_i (t
	- t_k), \quad t_k < t \le t_{k+1},
	\end{align} 
	where, $\alpha_i>0$,
	is a positive constant, which may be different for different
	agents. $\alpha_i$ is known as the {\em growth rate}\index{growth
		rate} for agent $i$. By writing $x_i(k)$ for the $i$th agent's
	share at the $k$th capacity event as $x_i(k) := x_i(t_k)$ we have
	\begin{align*}
	x_i(k+1) = \beta_i x_i(k) + \alpha_i T(k),
	\end{align*}
	where $
	T(k) := t_{k+1}-t_k,
	$
	is the time between events $k$ and $k+1$.  There are
	situations where not all agents may respond to every capacity
	event. Indeed, this is precisely the case considered in this paper. In
	this case agents respond asynchronously to a congestion notification
	and the AIMD model is easily extended by using our previous formalism
	by changing the multiplicative factor\index{multiplicative factor} to
	$\beta_i = 1$ at the capacity event if agent $i$ does not
	decrease.  
	
	\section{Divisible multi-resource allocation} \label{divisible_mul_res}
	
	%\subsection{Problem description}
	Let $\delta>0$ be a fixed constant; we denote by $\mathcal{F}_\delta$ the set of twice-continuously
	differentiable functions defined as follows 
	\begin{equation} \label{def_F_delta}
	%\begin{split}
	\Bigg\{ f:\mathbb R^m_+ \to \mathbb R
	\Big \lvert   \Big( x^j > 0 \implies   0 < \delta \nabla_j f(x) < x^j
	\text{ for all $j$}\Big) \text{ and } \nabla^2 f(x) \geq 0 \text{ for all }x\in \mathbb{R}^m_+ \Bigg\}.
	%\end{split}
	\end{equation}
	This is essentially the set of functions that are convex and increasing
	in each coordinate.
		We consider the problem of allocating $m$ divisible resources with capacity
	$C^j$, for $j\in \mathcal{M}$ among $n$
	competing agents, whose cost functions $f_1,\ldots,f_n$ belong to the
	set $\mathcal F_\delta$.  Each cost function is private and should be
	kept private. However, we assume that the set $\mathcal F_\delta$ is
	common knowledge: the control unit needs the knowledge of $\delta$ and the
	agents need to have cost functions from this set.
	
	Recall that $x^* = (x_1^{*1}, \ldots, x_n^{*m})$ is the solution of \eqref{obj_fn1}. We propose a distributed algorithm that determines instantaneous allocations $\{x_i^j(k)\}$, for all  $i, j$ and $k$. We also show empirically that
	for every agent $i$ and resource $R^j$, the long-term average allocations converge to the optimal allocations i.e.,
	\begin{align*}
	\overline{x}_i^j(k) \to {x}_i^{*j}
	\end{align*}
	as $k\to \infty$ (cf. \eqref{eq:longtermopt}) to achieve the
	minimum social cost.
	
	\subsection{Algorithm}
	
	Each agent runs a distinct distributed AIMD algorithm. We use $\alpha^j>0$
	to represent the additive increase factor or growth rate and $0 \leq \beta^j \leq 1$ to
	represent multiplicative decrease factor, both  corresponding to
	resource $R^j$, for $j \in \mathcal{M}$ and is uniform amongst all agents. 
	Every algorithm is initialized with the same set of parameters 
	$\Gamma^j$, $\alpha^j$, $\beta^j$ received from the control unit of the system. The constants $\Gamma^j$ is chosen based on the knowledge of fixed constant $\delta$ according to \eqref{gamma} to scale probabilities.
	We represent the one-bit \emph{capacity constraint event signals} by $S^j(k) \in \{0,
	1\}$ at time step $k$ for resource $R^j$,
	for all $j$ and $k$.
	At the start of the system the control unit initializes the capacity constraint event signals $S^j(0)$ with $0$, and updates $S^j(k)=1$ when the total allocation $\sum_{i=1}^n x_i^j(k)$ exceeds the capacity $C^j$ of a resource $R^j$ at a time step $k$. After each update, control unit broadcasts it to agents in the system signaling that the total demand has exceeded the capacity of the resource $R^j$. We describe the algorithm of control unit in Algorithm \ref{algoCU1}.
	\begin{algorithm}  \SetAlgoLined Input:
		$C^{j}$, for $j \in \mathcal{M}$.
		
		Output:
		$S^{j}(k)$, for $j \in \mathcal{M}$, $k \in
		\mathbb{N}$.
		
		Initialization: $S^{j}(0) \leftarrow 0$, for $j \in \mathcal{M}$,
		
		broadcast $\Gamma^{j} \leq \delta$;
		
		\ForEach{$k \in \mathbb{N}$}{
			
			\ForEach{$j \in \mathcal{M}$}{
				\uIf{ $\sum_{i=1}^{n} {x}_i^{j}(k) > C^{j}$}{
					$S^{j}(k+1) \leftarrow 1
					$\;
					
					broadcast $S^{j}(k+1)$;	
				}
				
				\Else{$S^{j}(k+1) \leftarrow 0$\;

		} }}
		
		\caption{Algorithm of control unit}
		\label{algoCU1}
	\end{algorithm}
	
	The algorithm of each agent works as follows.  At every time
	step, each algorithm updates its demand for resource $R^j$ in one of the
	following ways: an {\em additive increase (AI)} or a {\em multiplicative
		decrease (MD) phase}.
	In the additive increase phase, the algorithm increases its demand
	for resource $R^j$ linearly by the constant $\alpha^j$ until it
	receives a capacity constraint event signal $S^j(k) =1$ from the control unit of
	the system at time step $k$ that is,
	\begin{align*}
	x_i^j (k+1) = 	x_i^j (k) + \alpha^j.
	\end{align*}  
	After receiving the capacity constraint event signal $S^j(k) =1$ at the
	event of total demand exceeding the capacity of a resource $R^j$ (in multiplicative decrease phase), based on the probability $\lambda^j_{i}(k)$ each agent $i$ either responds to the capacity event with its updated demand for a resource $R^j$ or does not respond in the next time step with the goal that the resulting average
	allocation profile converge to the optimal allocations $x^*=(x_1^{*1},\ldots, x_n^{*m})$, for all
	agents and resources in the system. If $S^j(k) =1$, we thus have
	\begin{align*}  
	x_i^j(k+1)= \left\{
	\begin{array}{ll}
	\beta^j x_i^j(k) & \mbox{with probability } \lambda^j_{i}(k) , \\
	x_i^j(k) & \mbox{otherwise}.\\
	\end{array}
	\right.
	\end{align*}	
	The probability $\lambda^j_{i}(k)$
	depends on the average allocation and the derivative of cost function of agent $i$ with respect to $R^j$,
	for all $i$ and $j$.  It is calculated as
	follows
	\begin{align} \label{prob_x} \lambda^j_{i}(k) = \Gamma^j  
	\frac{{\nabla_j} f_i(\overline{x}^1_i(k), \overline{x}^2_i(k), \ldots,
		\overline{x}^m_i(k))}{\overline{x}^j_i(k)},
	\end{align}
	for all $i $, $j$ and $k$. This process repeats: after the reduction of consumption all agents can
	again start to increase their consumption until the next capacity event occurs.
	It is obviously required that always $ 0 < \lambda^j_{i}(k) < 1$. To
	this end the normalization factor $\Gamma^j$ is
	needed which is based on the set $\mathcal F_\delta$. The fixed constant $\delta >0$ is chosen such that $\Gamma^j$ satisfies the following
	 	\begin{align}\label{gamma_delta}
		0 <\Gamma^j \leq \delta, \text{ for all } j.  
		\end{align}
		At the beginning of
	the algorithm the normalization factor $\Gamma^j$ for
	resource $R^j$ is calculated explicitly as the following and broadcast to all agents in the system
	\begin{align}\label{gamma}
	\Gamma^j = 
	\inf_{x_1^1,\ldots,x_n^m \in \mathbb{R}_+,  f \in
		\mathcal{F}_\delta}
	\left(\frac{x^j}{\nabla_j f(x^1,
		x^2, \ldots, x^m)} \right), \text{ for all } j.  
	\end{align}
	 To capture the stochastic nature of the response to the capacity signal, we define the following Bernoulli random variables
	\begin{align} \label{bern_var}
	b^j_i(k)=
	\left\{
	\begin{array}{ll}
	1  & \mbox{with probability } \ \lambda^j_{i}(k),\\
	0 & \mbox{otherwise, } 
	\end{array}
	\right. 
	\end{align}
	for all $i$, $j$ and $k$.
	It is assumed that this set of random variables is independent.
	\begin{theorem} For a given $\delta >0$, if  $\overline{x}_i^j(k) >0$ and the cost function $f_i$ of agent $i$ belongs to $\mathcal{F}_\delta$, then for all $i, j$ and $k$, the probability $\lambda_i^j(k)$ satisfies $0 < \lambda_i^j(k) < 1$.
		\end{theorem}
		
		\begin{proof}
			Consider that $f_i \in \mathcal{F}_\delta$ and $\overline{x}_i^j(k) > 0$ for all $i$, $j$ and $k$ then from \eqref{def_F_delta}, we write that
			\begin{align} \label{eq117}
			 0 < \delta \nabla_j f_i(\overline{x}_i^1(k),\overline{x}_i^2(k), \ldots, \overline{x}_i^m(k)) < \overline{x}_i^j(k).
			\end{align}
			Given that $\overline{x}_i^j(k) >0$, dividing \eqref{eq117} by $\overline{x}_i^j(k)$ we obtain the following
			\begin{align} \label{eq101}
			0 <   \frac{\delta \nabla_j f_i(\overline{x}_i^1(k),\overline{x}_i^2(k), \ldots, \overline{x}_i^m(k))}{\overline{x}_i^j(k)} < 1, \text{ for all $i,j$ and $k$}.
			\end{align}
			We are aware that for a fixed constant $\delta>0$, the normalization factor $\Gamma^j$ satisfies $0 < \Gamma^j \leq \delta$,  for all $j$  (cf. \eqref{gamma_delta}). Therefore, placing $\Gamma^j$ in \eqref{eq101}, we obtain the following
			\begin{align} \label{eq114}
			0 <   \frac{\Gamma^j \nabla_j f_i(\overline{x}_i^1(k),\overline{x}_i^2(k), \ldots, \overline{x}_i^m(k))}{\overline{x}_i^j(k)} < 1, \text{ for all $i,j$ and $k$}.
			\end{align}
			Since, for all $i, j$ and $k$, an agent $i$ makes a decision to respond the capacity event of a resource $R^j$ with the probability $\lambda_i^j(k)$, which is mentioned as follows (cf.  \eqref{prob_x})
			\begin{align*}
			\lambda_i^j(k) = \Gamma^j \frac{ \nabla_j f_i(\overline{x}_i^1(k),\overline{x}_i^2(k), \ldots, \overline{x}_i^m(k))}{\overline{x}_i^j(k)} .
			\end{align*}
			Hence, after placing $\lambda_i^j(k)$ in \eqref{eq114}, we deduce that
			\begin{align*}
			0 <  \lambda_i^j(k) < 1,  \text{ for all $i,j$ and $k$}. 
			\end{align*} 
	\end{proof}	

	\begin{figure}[H]
		\centering
		\includegraphics[width=0.8\textwidth,clip=true,trim=7.5cm 8.9cm 2cm 4.8cm]{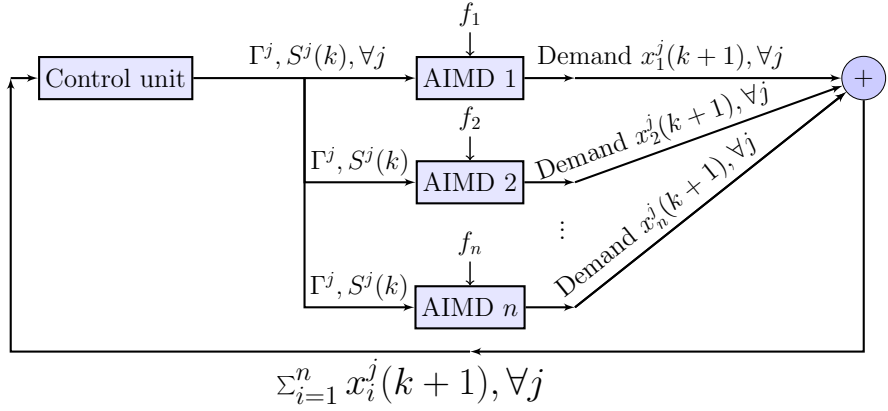}
		\caption{Block diagram of the proposed AIMD model}
		\label{Diag_AIMD}
	\end{figure}
	The system is described in Figure~\ref{Diag_AIMD} as a block diagram and the proposed distributed multi-resource allocation algorithm for each agent is
	described in Algorithm \ref{algo1}.
	
	\begin{algorithm}[H]  \SetAlgoLined Input:
		$S^{j}(k)$, for $j \in \mathcal{M}, k \in
		\mathbb{N}$
		and $\Gamma^j$,
		$\alpha^j, \beta^j$, for $j \in \mathcal{M}$.
		
		Output:
		$x^j_i(k+1)$, for $j \in \mathcal{M}$, $k \in
		\mathbb{N}$.
		
		Initialization: $x^j_i(0) \leftarrow 0$ and
		$\overline{x}^j_i(0) \leftarrow x^j_i(0)$, for
		$j \in \mathcal{M}$;

		\ForEach{$k \in \mathbb{N}$}{
			
			\ForEach{$j \in \mathcal{M}$}{

				\uIf{$S^{j}(k) = 1$}{
					$ \lambda^j_{i}(k) \leftarrow \Gamma^j
					\frac{{\nabla_j} f_i \left ( \overline{x}^1_i(k),
						\overline{x}^2_i(k), \ldots, \overline{x}^m_i(k)
						\right)}{\overline{x}^j_i(k)}$;
					
					generate independent Bernoulli random variable
					$b^j_i(k)$ with the parameter
					$\lambda^j_{i}(k)$;
					
					\uIf{ $b^j_i(k)=1$}{
						$x^j_i(k+1) \leftarrow \beta^j x^j_i(k)
						$;}
					
					\Else{ $x^j_i(k+1) \leftarrow x^j_i(k) $; }
					
				} \Else{
					$x^j_i(k+1) \leftarrow x^j_i(k) +
					\alpha^j$; }
				
				$\overline{x}^j_i(k+1) \leftarrow \frac{k+1}{k+2}
				\overline{x}^j_i(k) + \frac{1}{k+2} x^j_i(k+1);$
				
		} }
		
		\caption{Algorithm of agent $i$ (AIMD $i$) }
		\label{algo1}
	\end{algorithm}
	We observe using the Experiment \ref{results} that the average
	allocation of resource $\overline{x}_i^j(k)$ results in the optimal
	value $x_i^{*j}$ over time, for all $i$ and $j$.
	
	\begin{remark} Suppose there are $m$ resources in
		the system, then communication overhead will be
		$\sum_{j=1}^{m} S^{j}(k)$ bits at $k$ time step, for all
		$k$. In the worst case scenario this will be $m$ bits per
		time unit, which is quite low. Furthermore, the communication complexity
		does not depend on the number of agents in the system.
	\end{remark}
	
	\subsection{Experiments} \label{results}
	For convenience we use only two resources
	$R^1$ and $R^2$ in the experiment. We denote $x_i^1(k)$ as
	allocation of $R^1$ and $x_i^2(k)$ as
	allocation of $R^2$
	of agent $i$ at time step $k$, for all agents competing for the resources in the
	system.  We chose $60$ agents and the normalization factors
	$\Gamma^1 = \Gamma^2 = 1/35$. The
	initial states of all agents for $R^1$ and $R^2$ resources are initialized with $0$. For resource $R^1$ we chose the additive increase factor $\alpha^1=0.01$ and
	multiplicative decrease factor $\beta^1=0.85$
	and for resource $R^2$ we chose $\alpha^2 = 0.012$ and $\beta^2 = 0.80$, respectively. The resource capacities $C^1 = 15$
	and $C^2 = 20$ are of $R^1$ and $R^2$ resources, respectively. Suppose, for all $i$, $a_i \in \{1, 2, \ldots, 25\}$ and
		$b_i \in \{1, 2, \ldots, 10\}$ are uniformly distributed random variables. Using the random variables $a_i$ and $b_i$ we consider the following cost function to generate random costs of each agent at different time steps 
	\begin{align*}
	f_{i}(x_i^1, x_i^2) = \frac{1}{2}a_i(x_i^1)^2 + \frac{1}{2}a_i(x_i^2)^2 +
	\frac{1}{4}b_i(x_i^1)^4 + \frac{1}{4}b_i(x_i^2)^4. %\label{obj_func_1a}
	\end{align*}
	\begin{figure}
		\centering
		\begin{minipage}{0.45\textwidth}
			\centering
			\includegraphics[width=1\linewidth]{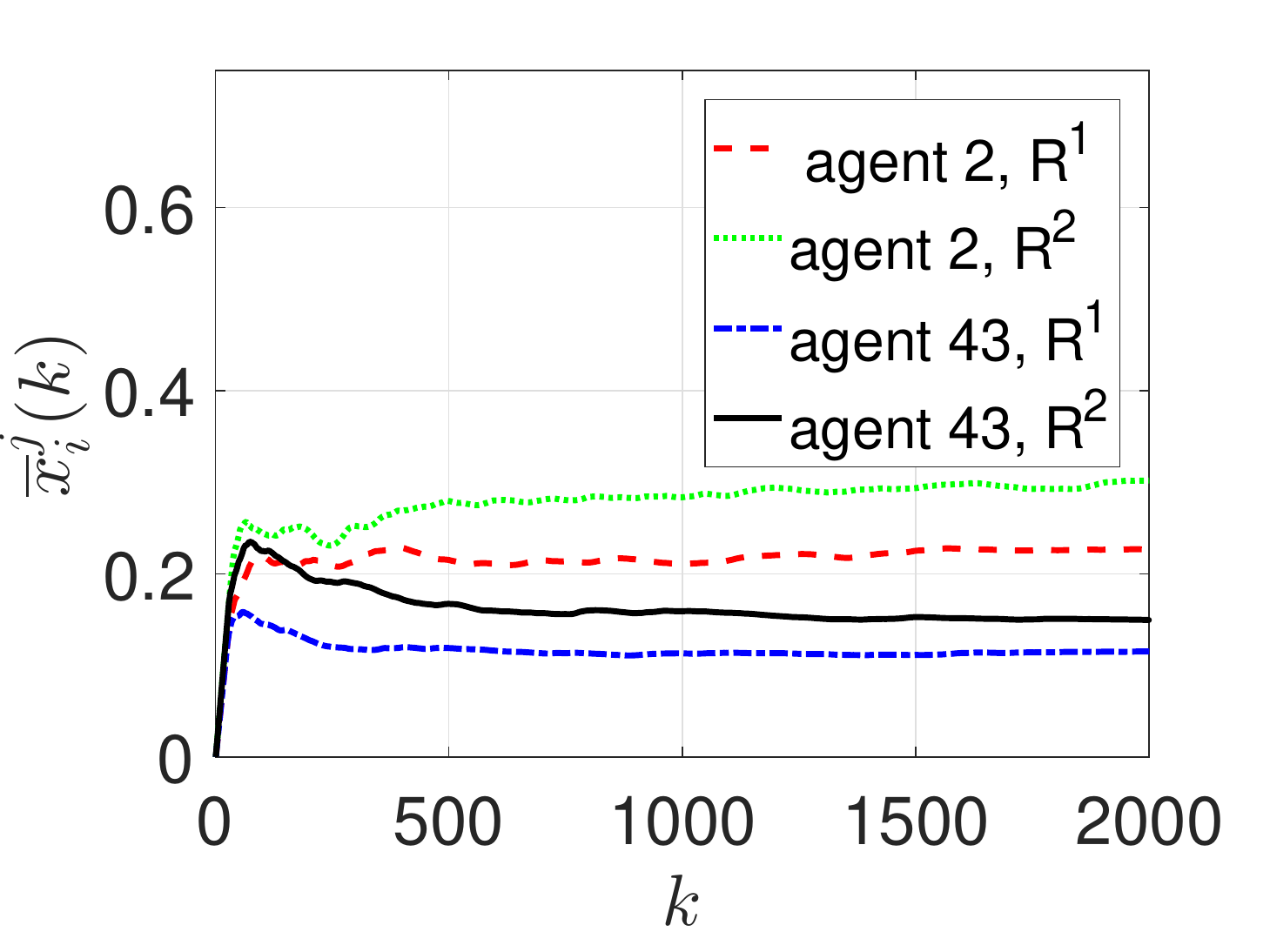}
			\captionof{figure}{Evolution of average allocation of\\ resources}
			\label{avg_AIMD}
		\end{minipage}%
		\begin{minipage}{0.45\textwidth}
			\centering
			\includegraphics[width=1\linewidth]{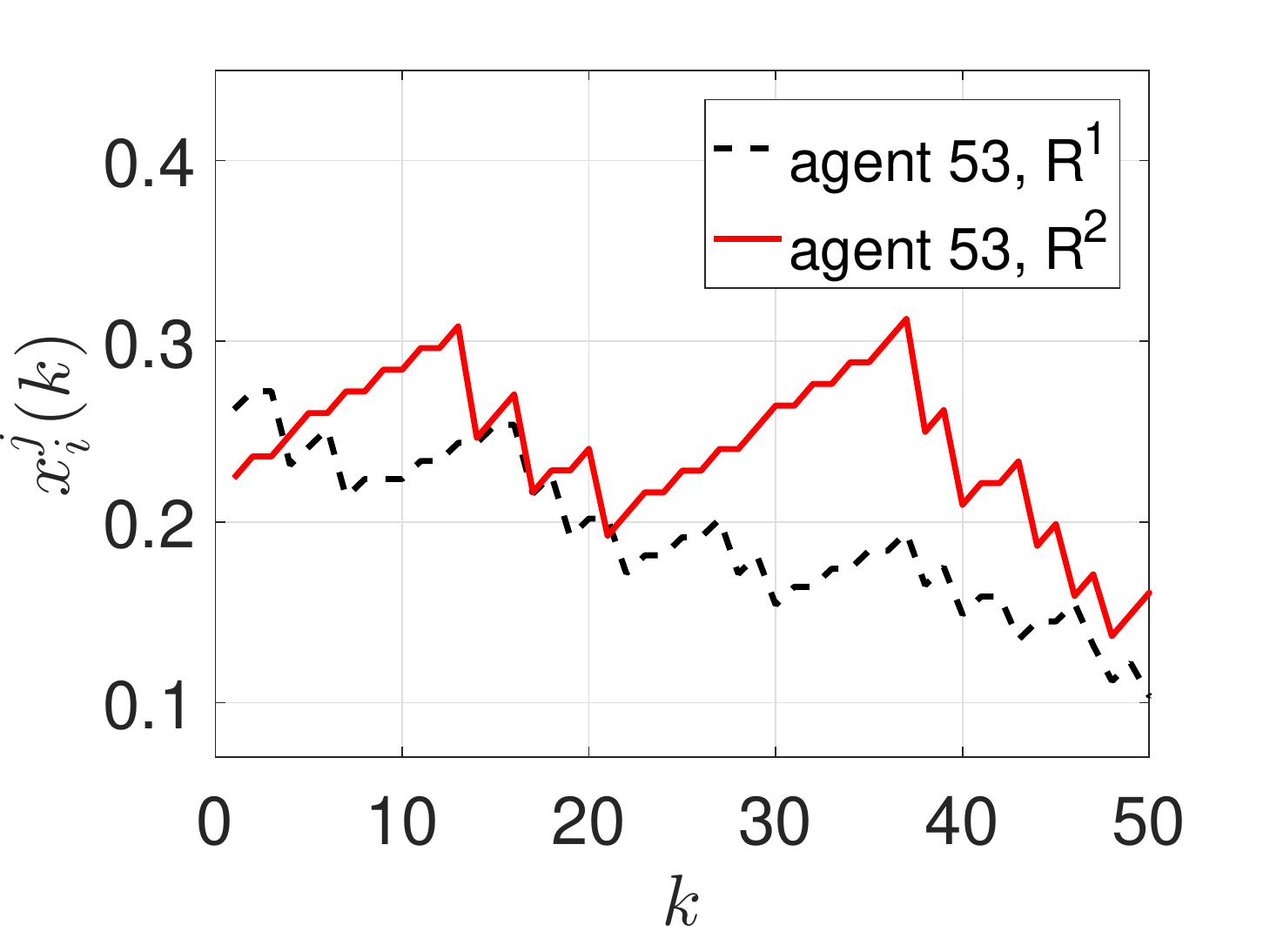}
			\captionof{figure}{Instantaneous allocation of resources for last $50$ time steps}
			\label{instant_alloc}
		\end{minipage}
	\end{figure}
	The following are some of the
	results obtained from the experiment. We select agents randomly to plot the figures, and mention
	the legend wherever necessary. It is observed in Figure (\ref{avg_AIMD}) that the average allocations $\overline{x}_i^1(k)$  and $\overline{x}_i^2(k)$ converge over time to their respective optimal values $x^{*1}_i$ and $x^{*2}_i$, for all $i$. 
	The allocation phases (AI and MD) are demonstrated in Figure (\ref{instant_alloc}) that shows the instantaneous allocation $x_i^1(k)$ and $x_i^2(k)$ over last $50$ time steps. 
	
	As we know that, to achieve optimality the derivatives of the
	cost functions of agents for a particular resource should make a consensus. Figure (\ref{err_grad}) is the error bar of derivatives $\nabla_1f_i$ and $\nabla_2f_i$ of cost functions $f_i$ for single simulation calculated across all agents. It illustrates that the derivatives of cost functions of all agents with respect to a particular resource  concentrate more and more over time around the same value. 
	Hence, the long-term average allocation of resources for the
	stated optimization problem is optimal.
	\begin{figure}
		\centering
		\begin{minipage}{0.45\textwidth}
			\centering
			\includegraphics[width=1\linewidth]{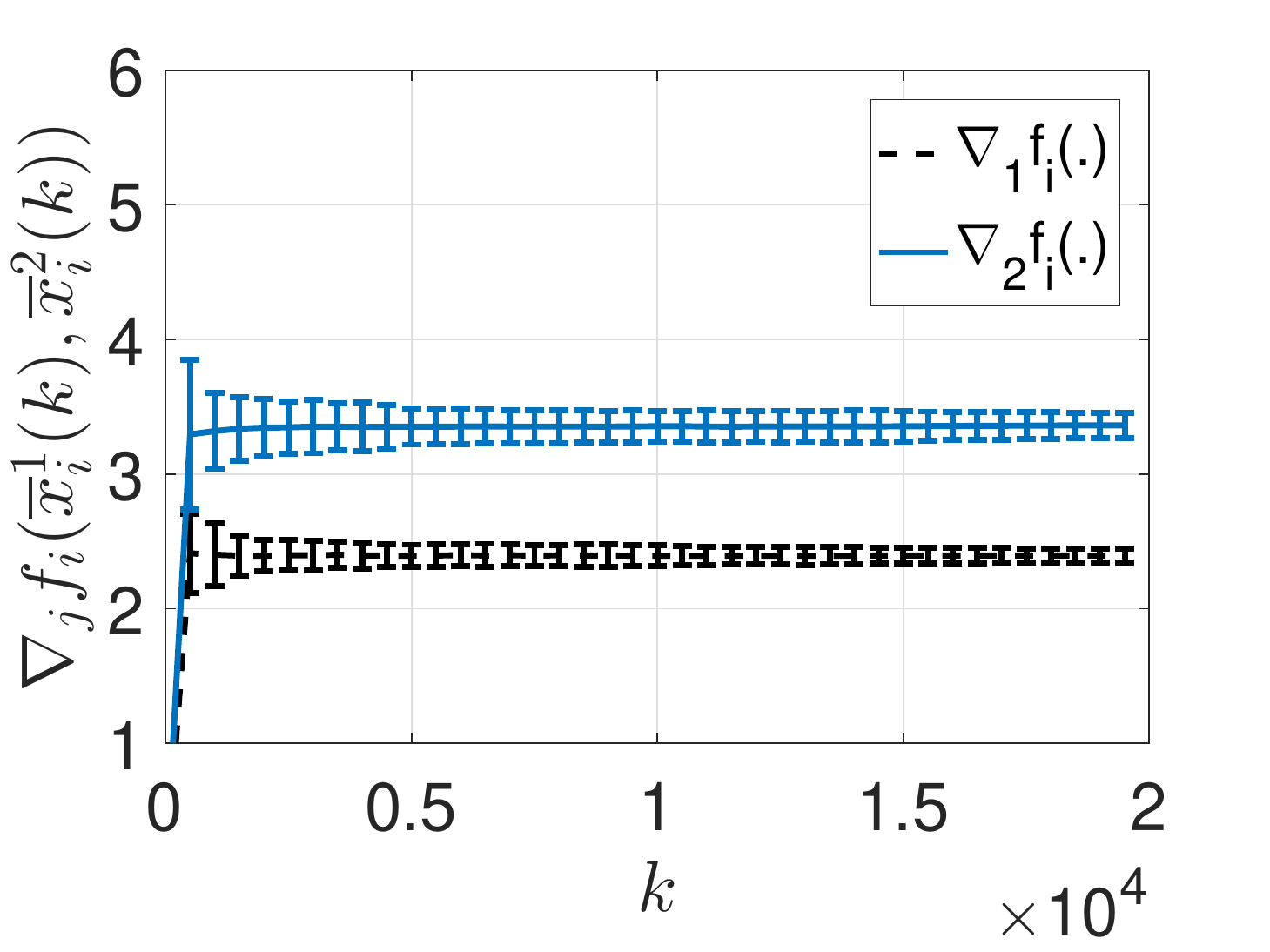}
			\captionof{figure}{Evolution of profile of derivatives
				\\ of $f_i$ of all agents}
			\label{err_grad}
		\end{minipage}%
		\begin{minipage}{0.45\textwidth}
			\centering
			\includegraphics[width=1\linewidth]{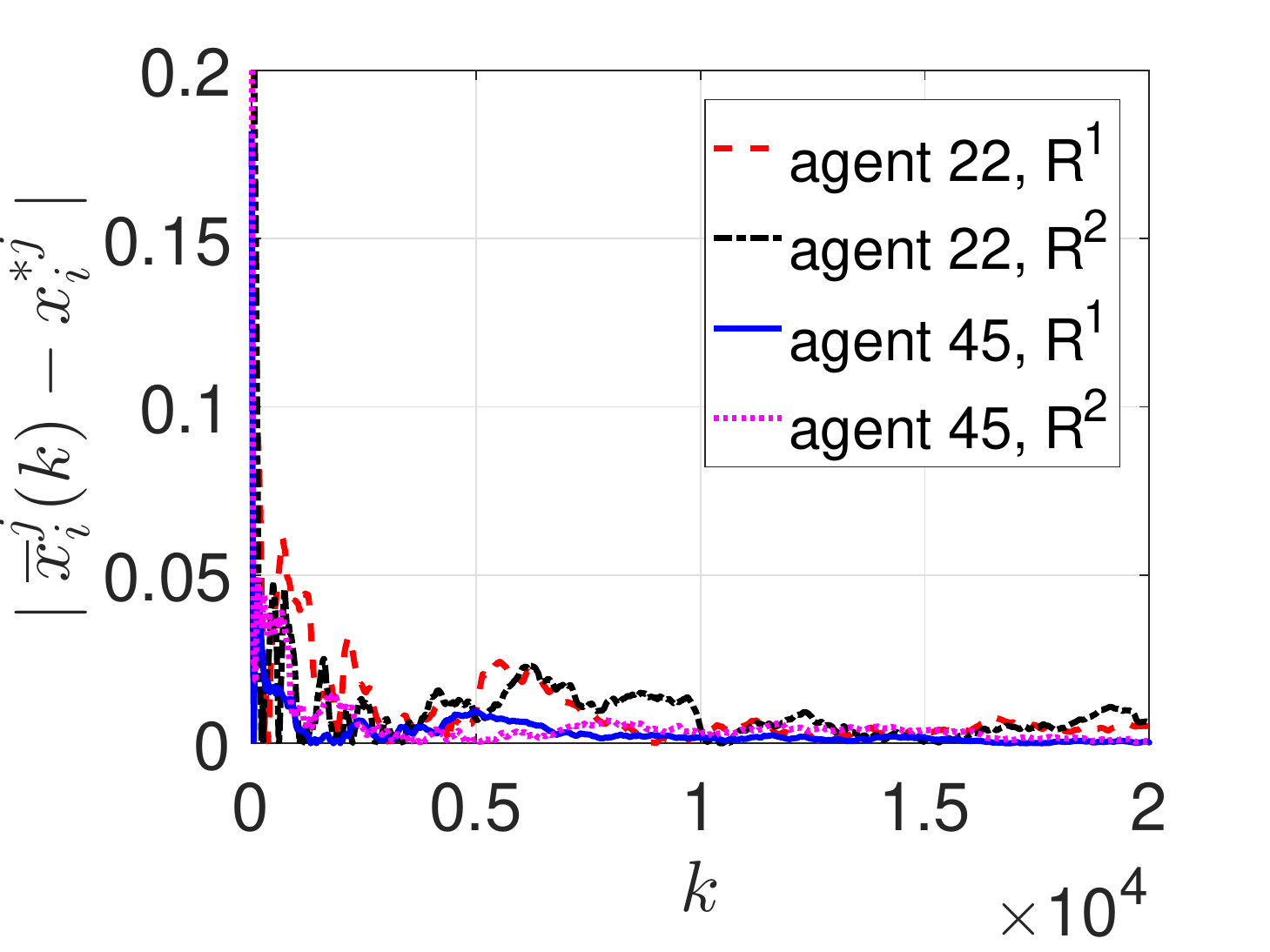}
			\captionof{figure}{Evolution of absolute difference between average allocation and the optimal allocation (calculated)}
			\label{abs_avgAIMD_Ins}
		\end{minipage}
	\end{figure}	
	To validate the results received from the algorithm, we calculate the optimal values $x_i^{*1}$ and $x_i^{*2}$ using the interior-point method for the same optimization problem, for all $i$. Let $K$ be the largest time step used in simulation, then Figure (\ref{abs_avgAIMD_Ins}) illustrates that the long-term average allocation  $\overline{x}_i^1(K)$ is approximately same as calculated optimal value $x_i^{*1}$ and similarly $\overline{x}_i^2(K) \approx x_i^{*2}$. It can be further seen in Figure (\ref{frac_func_optAIMD_Ins}) that the ratio of $\sum_{i=1}^{n} f_i(\overline{x}_i^1(K), \overline{x}_i^2(K))$ and $\sum_{i=1}^{n} f_i(x_i^{*1}, x_i^{*2})$ is also close to $1$.
	\begin{figure}
		\centering
		\begin{minipage}{0.45\textwidth}
			\centering
			\includegraphics[width=1\linewidth]{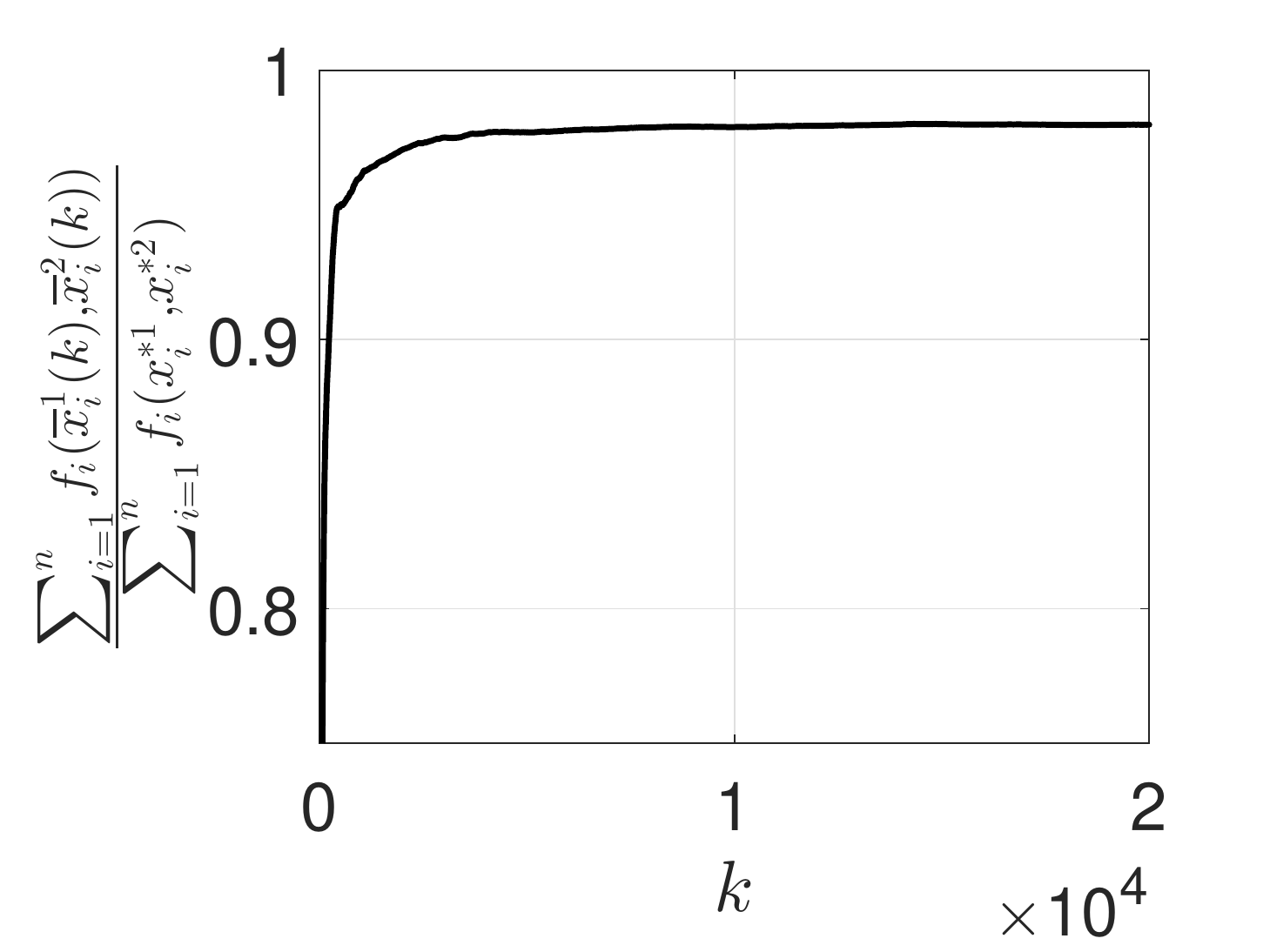}
			\captionof{figure}{Evolution of $\frac{\sum_{i=1}^{n} f_i(\overline{x}_i^1(k), \overline{x}_i^2(k))}{\sum_{i=1}^{n} f_i(x_i^{*1}, x_i^{*2})}$}
			\label{frac_func_optAIMD_Ins}
		\end{minipage}%
		\begin{minipage}{0.45\textwidth}
			\centering
			\includegraphics[width=1\linewidth]{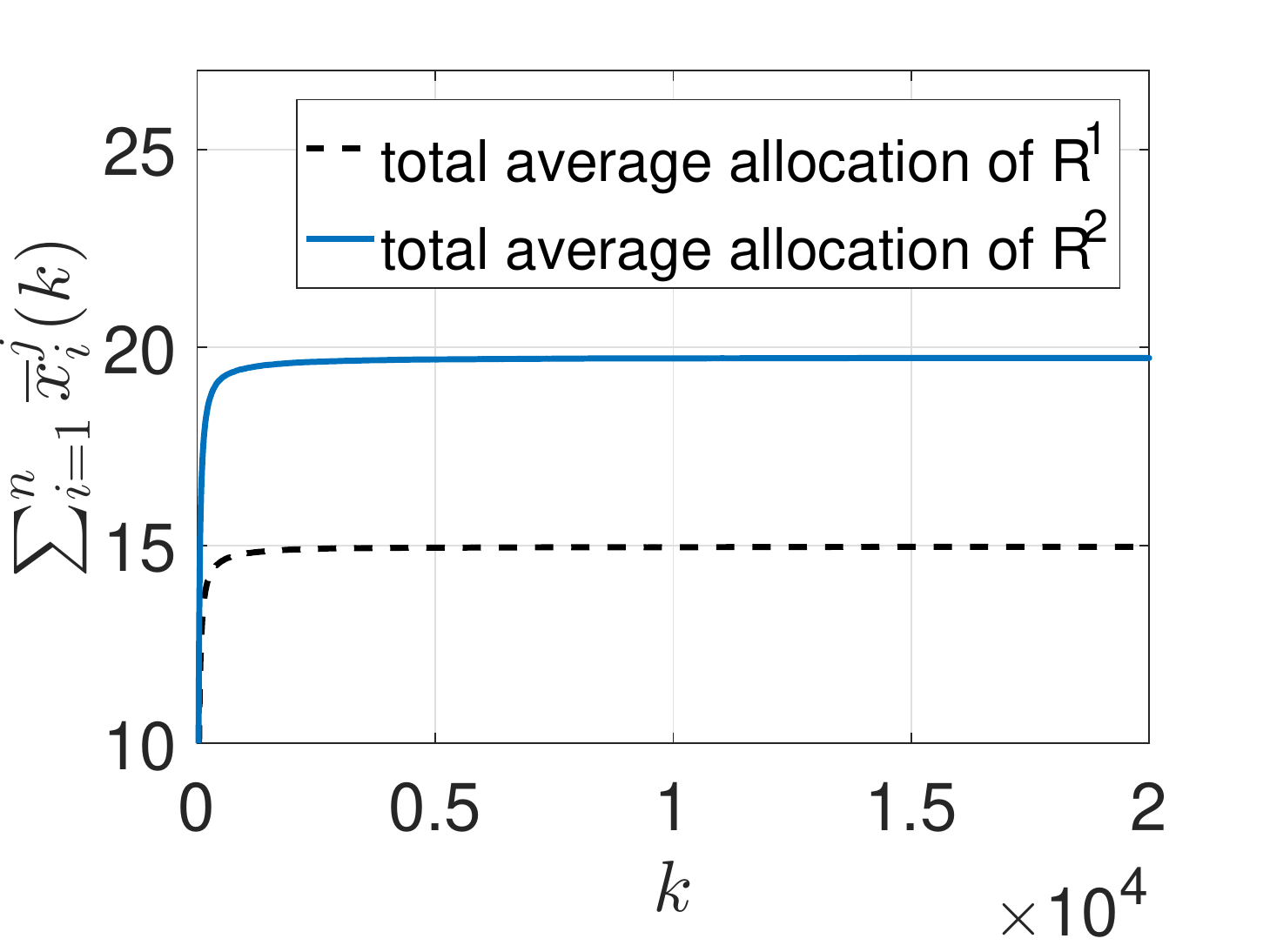}
			\captionof{figure}{Evolution of sum of average allocation of resources, the capacities are $C^1=15$ and $C^2=20$.}
			\label{sum_avg}
		\end{minipage}
	\end{figure}

	\begin{figure}
		\centering
		\includegraphics[width=2.9in]{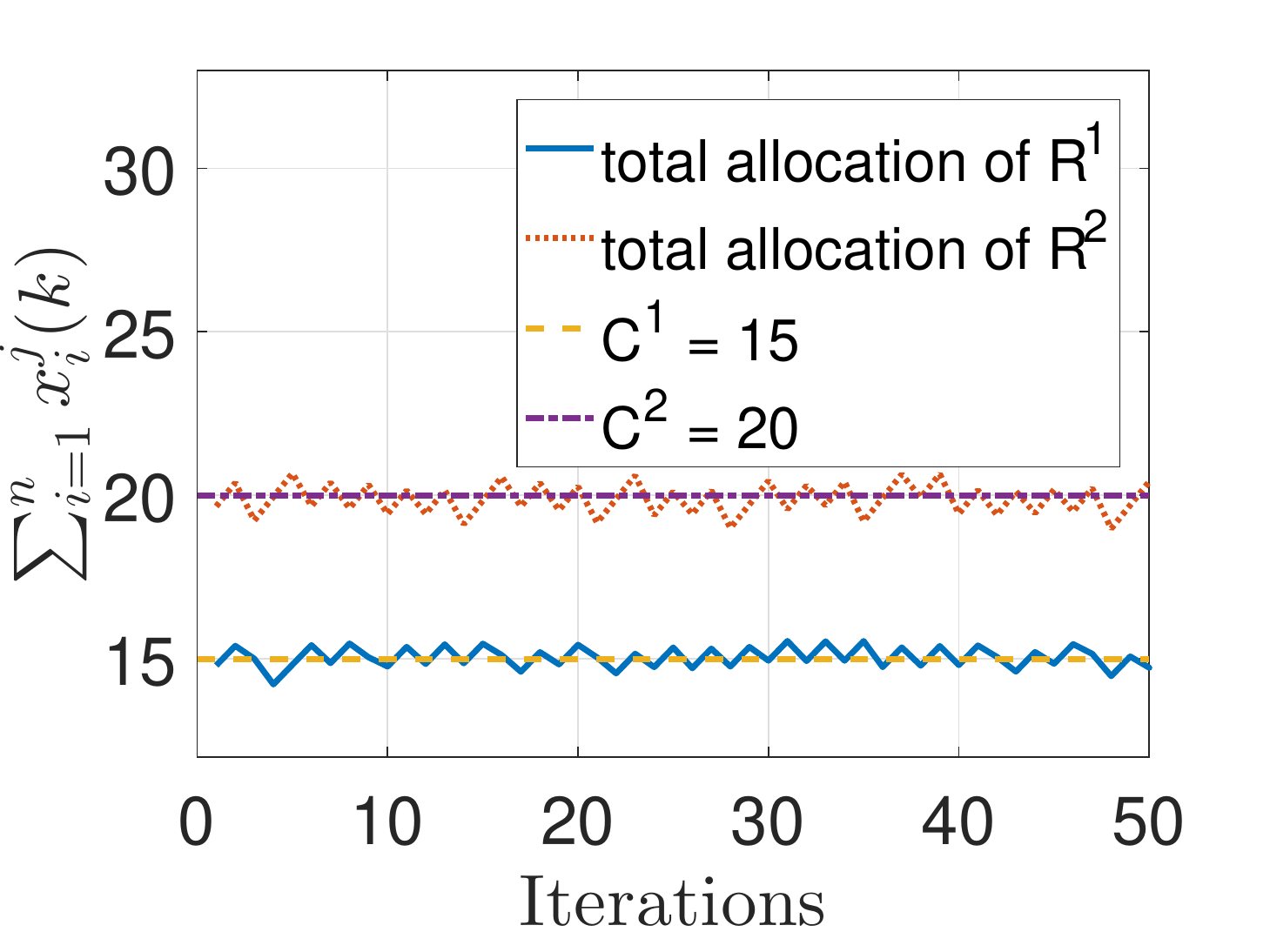} 
		\caption{Total allocation of resources for last $50$ time steps} 
		\label{sum_alloc}
	\end{figure}
		Figure (\ref{sum_avg}) illustrates the sum of average allocations $\sum_{i=1}^{n} \overline{x}_i^j(K)$ over time, we observe that it is approximately equal to the respective capacities, for all resources $R^j$. Figure (\ref{sum_alloc}) shows the sum of instantaneous allocations
	for last $50$ time steps of $R^1$
	and $R^2$ resources with capacities $C^1 = 15$ and $C^2 = 20$, respectively, we observe that the sum of instantaneous allocations are concentrated around the respective capacities. To overcome the overshoots of total allocations of resource $R^j$, we suppose $\gamma^j < 1$ and modify the algorithm of control unit to broadcast the capacity constraint event signal $S^j(k)=1$ when $\sum_{i=1}^{n} x_i^j(k) > \gamma^j C^j$, for all $j$ and $k$.
	
	\section{ Binary multi-resource allocation} \label{bin_imp} In contrast to divisible resources, indivisible unit-demand resources or the binary resources are either allocated one unit to an agent or not allocated, therefore algorithms proposed in Section \ref{divisible_mul_res} do not work in this setting. Hence, we propose different distributed algorithms which suit such resource allocation problems. We name these distributed algorithms as {\em binary multi-resource allocation} algorithms which are loose variants of the algorithms of Section \ref{divisible_mul_res}. We find many applications or cases where there is need to solve binary multi-resource
	allocation optimization problem in a distributed fashion, e.g.,
	the allocation of parking spaces for different type of cars, say
	the electric cars near the charging points and the conventional
	cars in the allotted parking spaces. A parking space is either allocated to a user or not
	allocated. In this section, we use the same notations as introduced earlier until stated
	otherwise.
	
	Suppose $n$ agents are competing for $m$ indivisible unit-demand
	resource $R^1, R^2, \ldots, R^m$ which have
	capacities $C^1, C^2, \ldots, C^m$, respectively. Let $g_1, g_2, \ldots, g_n$ be the cost functions of agents, we consider that agents do not share their cost functions or
	allocation information with other agents.  For fixed $i$, $j$ and $k$, let $\xi_i^j(k)$ be the Bernoulli random variable which denotes whether agent $i$ receives one unit of
	resource $R^j$ at time step $k$ or not. Further, let $y_i^j \in [0, 1]$ be the average allocation of indivisible unit-demand resource $R^j$ of agent $i$, which is calculated as follows
	\begin{align} \label{average_eqn2}
	{y}_i^j(k) = \frac{1}{k+1} \sum_{\ell=0}^k \xi^j_i(\ell),
	\end{align}
	for $i \in \{1,\ldots, n\}$ and $j \in \{1,\ldots, m\}$. Notice that \eqref{average_eqn2} is different from \eqref{average_eqn} in the sense that it calculates average using the Bernoulli random variable $\xi^j_i$ for indivisible unit-demand resource $R^j$, whereas \eqref{average_eqn2} calculates average using real valued $x_i^j$ for divisible resource, for all $j$. Here, we consider that all resources are utilized on average.
	
	\begin{comment}
	If $g: \mathbb{R}_+^m \rightarrow \mathbb{R}$, then similar to $\mathcal F_\delta$ of  \eqref{def_F_delta}, we define $\mathcal G_\delta$ as set of second order continuously differentiable, convex and increasing functions, where $\delta>0$ is a fixed constant (cf. \eqref{tau}). 
	\end{comment}
	%\begin{comment}
	\begin{comment}
	Similar to \eqref{def_F_delta}, for a fixed constant $\delta>0$; we define $\mathcal G_{\delta}$ as a set of second order continuously differentiable functions as follows
	\begin{align} \label{G_delta}
	\Bigg\{ g:\mathbb R_+^m \to \mathbb R \Big \lvert \mbox{ for all } y\in \mathbb [0,1]^m; \Big( \nabla_j g(y) >0 \implies 0 < {\delta}y^j < \nabla_j g(y) \text{ for all } j  \Big) \mbox{ and } \nabla^2 g(y) \geq 0   \Bigg\},
	\end{align}
	which is basically the set of convex and increasing functions. 
		\end{comment}
		
	We consider that there exists $\delta > 0$ such that $\mathcal G_{\delta}$ is a set of second order continuously differentiable, convex and increasing functions similar to \eqref{def_F_delta} and $g_1, g_2, \ldots, g_n \in \mathcal G_{\delta}$. In contrast to Section \ref{divisible_mul_res} we do not define $\mathcal G_{\delta}$ or show how to construct it here. Moreover, we assume that $\mathcal G_{\delta}$ is common knowledge to the control unit and the cost function $g_i$ depends on $y_i^j$, for all $i$ and $j$. Then, instead of defining the resource allocation problem in terms of the instantaneous allocations $\{\xi_i^j(k)\}$, we consider an objective and constraints defined in terms of averages (cf. \eqref{obj_fn1})
	\begin{align}
	\begin{split}\label{1139}
	\min_{y_1^1, \ldots, y_n^m} \quad &\sum_{i=1}^{n} g_i(y_i^1,\ldots,y_i^m),
	\\ \mbox{subject to } \quad &\sum_{i=1}^{n} y_i^j  =  C^j, \quad j \in \{1, \ldots, m\}, 
	\\  &y_i^j\geq 0, \quad i \in \{1, \ldots, n\}, \ j \in \{1, \ldots, m\}.
	\end{split}
	\end{align}
	Let $y^* = ({y}_1^{*1}, \ldots, {y}_n^{*m}) \in \mathbb{R}_+^{nm}$ denotes the solution to \eqref{1139}.
	Next, we propose a distributed algorithm that determines instantaneous allocations $\{\xi_i^j(k)\}$ and show empirically that for every agent $i$ and resource $R^j$, the long-term average allocations converge to the optimal allocations
	\begin{align*}
	\frac{1}{k+1} \sum_{\ell=0}^{k} \xi_i^j(\ell) \to {y}_i^{*j} \quad \mbox{(in distribution)}
	\end{align*}
	as $k\to \infty$, thereby achieving the
	minimum social cost in the sense of long-term averages.
	
	\subsection{Algorithm}
	The distributed binary multi-resource allocation algorithm is run by each agent in the system. Let $\tau^j  \in \mathbb{R}_+$ be the gain parameter, $\Omega^j(k) \in \mathbb{R}_+$ be the normalization factor and $C^j$ be the capacity of resource $R^j$, for all $j$. The control unit updates $\Omega^j(k)$ according to \eqref{omega} at each time step and broadcasts it to all agents in the system, for all $j$ and $k$.  When an agent
	joins the system at a time step $k$; it receives the set of parameters
	$\Omega^j(k)$, $\tau^j$ and $C^j$ for resource $R^j$, for all
	$j$. Each algorithm updates its resource
	demand over time either by demanding one unit of the resource or
	not demanding it.
	
	The normalization factor $\Omega^j(k)$ depends on its earlier value, $\tau^j$, the capacity $C^j$ and total utilization of resource $R^j$ at earlier time
	step, for all $j$ and $k$. After receiving this signal the
	agent $i$'s algorithm responds in a probabilistic manner. It
	calculates its probability $\sigma^j_i(k)$ using its
	average allocation of resource $R^j$ and the derivative
	of its cost function, for all $j$ and $k$, as mentioned in \eqref{prob_x2}. Using this probability it finds out the Bernoulli
	distribution, similar to \eqref{bern_var}, based on the outcome of the resulting random
	variables $0$ or $1$, the algorithm decides whether to
	demand one unit of the resource or not. If the random variable takes
	value $1$, then the algorithm demands one unit of the resource
	otherwise it does not demand that resource. This process repeats over time. The block diagram of the system is described in Figure~\ref{Diag_BAIMD}. We describe the proposed {\em binary multi-resource allocation} algorithm for the control unit in Algorithm \ref{algoCU2} and for each agent in Algorithm \ref{algo3}.
	\begin{figure}[H]
		\centering
		\includegraphics[width=0.8\textwidth,clip=true,trim=7.5cm 8.9cm 2cm 4.5cm]{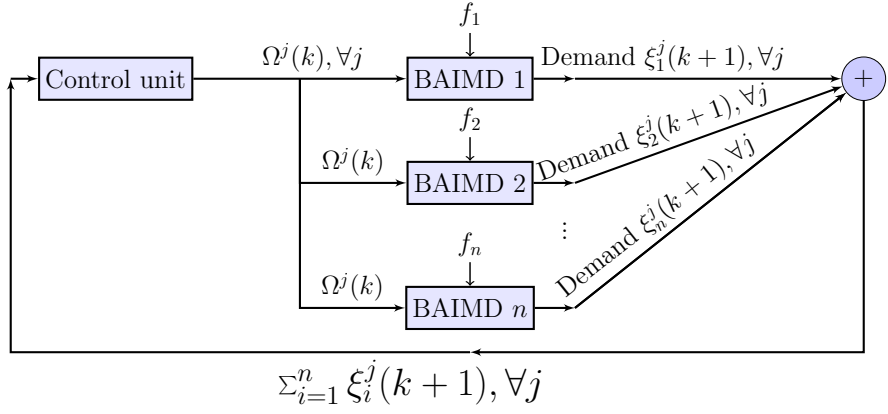}
		\caption{Block diagram of the proposed Binary resource allocation model}
		\label{Diag_BAIMD}
	\end{figure}
	\begin{algorithm}[H]  \SetAlgoLined Input:
		$C^{j}$, $\tau^j$ for $j \in \mathcal{M}$.
		
		Output:
		$\Omega^{j}(k)$, for $j \in \mathcal{M}$, $k \in \mathbb{N}$.
		
		Initialization: $\Omega^{j}(0) \leftarrow fixed\_value$, for $j \in \mathcal{M}$,
		
		\ForEach{$k \in \mathbb{N}$}{
			
			\ForEach{$j \in \mathcal{M}$}{
				
				calculate $\Omega^j(k+1)$ according to \eqref{omega} and broadcast in the system;	
		} }
		\caption{Algorithm of control unit}
		\label{algoCU2}
	\end{algorithm}
	
	\begin{algorithm}[H]  \SetAlgoLined Input:
		$\Omega^j(k)$, for $k \in \mathbb{N}, j \in
		\mathcal{M}$.
		
		Output: $\xi^j_i(k+1)$, for $j \in \mathcal{M}$ and
		$k \in \mathbb{N}$.
		
		Initialization: $\xi^j_i(0) \leftarrow 1$ and
		${y}^j_i(0) \leftarrow \xi^j_i(0)$, for
		$ j \in \mathcal{M}$.
		
		\ForEach{$k \in \mathbb{N} $}{

			\ForEach{$j \in \mathcal{M}$}{
				$\sigma^j_i(k) \leftarrow \Omega^j(k)
				\frac{{y}^j_i(k)}{ \nabla_j
					{g_i({y}_i^1(k)}, \ldots,
					{y}^m_i(k))}$; 
				
				generate Bernoulli independent random variable
				$b^j_i(k)$ with the parameter $\sigma^j_i(k)$;
				
				\eIf{ $b^j_i(k) = 1$}{
					$\xi^j_i(k+1) \leftarrow 1$; }
				{$\xi^j_i(k+1) \leftarrow 0$; }

				${y}^j_i(k+1) \leftarrow \frac{k+1}{k+2}
				{y}^j_i(k) + \frac{1}{k+2} \xi^j_i(k+1);$} 
		}
		\caption{Binary multi-resource allocation algorithm of agent $i$}
		\label{algo3}
	\end{algorithm}	
	The proposed algorithms are different from that of Section \ref{divisible_mul_res}. In Algorithm \ref{algoCU2}, the control unit broadcasts $\Omega^j(k) \in \mathbb{R}_+$ at each time step $k$, unlike one bit capacity event signal $S^j(k) = 1$ in Algorithm \ref{algoCU1} (after total demand exceeds the capacity of resource $R^j$). The normalization factor $\Omega^j(k)$ is different from $\Gamma^j$ as well, in the manner that $\Omega^j(k)$ is updated at every time step by control unit based on the total utilization and the capacity of the resource $R^j$. The control unit broadcasts it to all agents in the system, whereas $\Gamma^j$ is a fixed quantity calculated by control unit based on the allocations and derivatives of cost functions of agents for each resource $R^j$. An agent receives $\Gamma^j$ from control unit once when it joins the system. Apart from these, the way of calculating the probability $\sigma_i^j(k)$ (cf. \eqref{prob_x2}) with which an agent decides about its resource demand is also quite different from  \eqref{prob_x}. To calculate $\sigma_i^j(k)$, $\Omega^j(k)$ is used which varies with time, in contrast to fixed value $\Gamma^j$ used in \eqref{prob_x} and also the ratio of ${y}_i^j$ and $\nabla_j g_i(.)$ used here is reciprocal of \eqref{prob_x}.

After introducing the algorithms, we describe here how to calculate different factors. The control unit updates the normalization factor $\Omega^j(k + 1)$ in the following manner for all $j$ and $k$ using the parameters and the common knowledge of $\mathcal{G}_\delta$ and broadcasts it to all agents in the system
\begin{align} \label{omega}
\begin{split}
\Omega^j(k+1) = \Omega^j(k) -  \tau^j  \left (\sum_{i=1}^n \xi^j_i(k) -C^j \right ).
\end{split}
\end{align}
To keep the system stable we consider \emph{gain parameter} $\tau^j$ which is
derived from the expected utilization of a resource $R^j$ (similar to \cite{Griggs2016}), thus
%\begin{comment}
\begin{align*} %\label{tau}
\tau^j = \inf_{{y}_1^1, \ldots, {y}_n^m \in \mathbb{R}_+, g_1, \ldots, g_n \in \mathcal{G}_\delta} \left ( \sum_{i=1}^{n} \frac{ y_i^j}{\nabla_j{g_i({y}_i^1, {y}_i^2, \ldots, {y}_i^m)}}  \right )^{-1}, \text{ for all $j$.} 
\end{align*}
	When an agent $i$ receives $\Omega^j(k)$ from the control unit at time step $k$, it calculates the probability $\sigma_i^j(k)$  in the following manner to make a decision about its demand for resource
		$R^j$ at next time step
	\begin{align} \label{prob_x2}
	\sigma_i^j(k) =  \Omega^j(k) \frac{ y^j_i(k)}{ \nabla_j{g_i({y}_i^1(k), {y}_i^2(k), \ldots, {y}_i^m(k))}}, \text{ for all $i$, $j$ and $k$}.
	\end{align}	
	Here, $\Omega^j(k)$ is used to bound the
	probability $\sigma_i^j(k) \in (0,1)$, for all $i$, $j$ and $k$.

	Suppose that $\Omega^j(k)$ takes the floating point values
	represented by $\mu$ bits. If there are $m$ indivisible
	unit-demand resources in the system, then the communication
	overhead in the system will be $\mu m$ bits per time
	unit. This is in contrast to the divisible resource allocation
	problem, where in the worst case scenario only $m$ bits are
	required per time unit. The communication complexity in this
	case also is independent of the number of agents in the
	system.
	\subsection{Experiments} \label{bin_exp}
	
	In this experiment, for convenience we used two resources
	$R^1$ and $R^2$ and three cost functions. Each cost function represents a class, and a set of
	agents belong to each class. The cost depends on the average allocation of indivisible unit-demand resources. For agent $i$, we consider the following cost functions:
	
	\begin{equation*} \label{bin_func} g_{i}(y_i^1,y_i^2)= \left\{
	\begin{array}{ll}
	(y_i^1+y_i^2)^2/2 + (y_i^1+y_i^2)^4/8 & \mbox{w. p.} \ 1/3, \\
	(y_i^1+y_i^2)^4/4 + (y_i^1+y_i^2)^6/12   & \mbox{w. p.} \ 1/3, \\
	(y_i^1+y_i^2)^6/6 + (y_i^1+y_i^2)^8/16 & \mbox{w. p.} \ 1/3.
	\end{array}
	\right.
	\end{equation*}
	
	We consider $900$ agents competing for
	indivisible unit-demand resources in the system. Along with
	this, we chose the capacity of resource
	$R^1$ as $C^1 = 450$ and that of $R^2$
	as $C^2 = 350$. The parameters are initialized with the following values, $\Omega^1(0) = 0.350 $,
	$\Omega^2(0) = 0.328$, $\tau^1 = 0.0002275$ and
	$\tau^2= 0.0002125$.  In the experiment, for the sake of simplicity we assume that all agents join the system at the start of the algorithm. We classified the agents as follows; agents $1$ to $300$
	belong to class $1$, agents $301$ to $600$ belong to class $2$
	and agents $601$ to $900$ belong to class $3$. 	In the experiment we observed that a few times $\sigma_i^j(k)$ overshoots $1$, to overcome it we use $\sigma_i^j(k) = \min \Big\{\Omega^j(k) \frac{ y^j_i(k)}{ \nabla_j{g_i({y}_i^1(k), {y}_i^2(k), \ldots, {y}_i^m(k))}}, 1 \Big\}$, for all $i, j$ and $k$.
		\begin{figure}[H]
		\centering
		\includegraphics[width=2.8in]{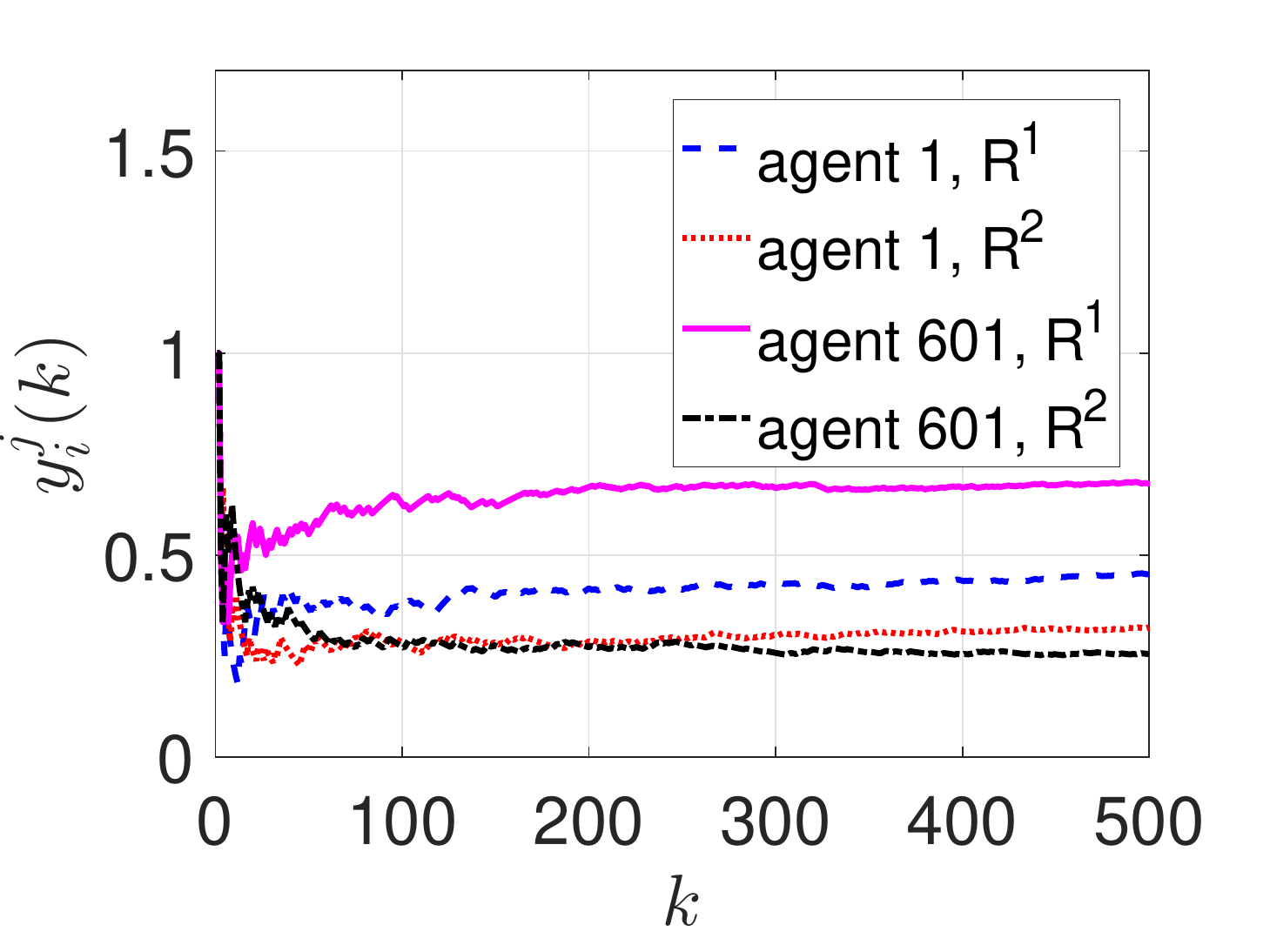} 
		\caption{Evolution of average allocation of resources} 
		\label{avg_BAIMD}
	\end{figure}
	We illustrate some of the results of the experiment here. Figure (\ref{avg_BAIMD}) shows that the
	average allocations ${y}_i^1(k)$ and ${y}_i^2(k)$ converge over
	time to their respective optimal values $y_i^{*1}$ and $y_i^{*2}$, respectively for agent $i$.
	As mentioned earlier in \eqref{optimality}, to show the optimality of a solution, the derivative of cost function of all agents with respect to a particular resource should make a consensus. Since, the derivatives $\nabla_1g_i$ and $\nabla_2g_i$ are same, for all $i$, we mention here just one of these derivatives. Figure (\ref{err_grad_BAIMD}) illustrates the profile of derivatives of cost functions of all agents for a single simulation with respect to resource $R^1$, we observe that they converge with time and hence make a consensus.
	The empirical results thus obtained, show the convergence of the long-term average allocation of resources to their respective optimal values using the consensus of derivatives of the cost functions. 
	\begin{figure}[H]
		\centering
		\begin{minipage}{0.45\textwidth}
			\centering
			\includegraphics[width=1\linewidth]{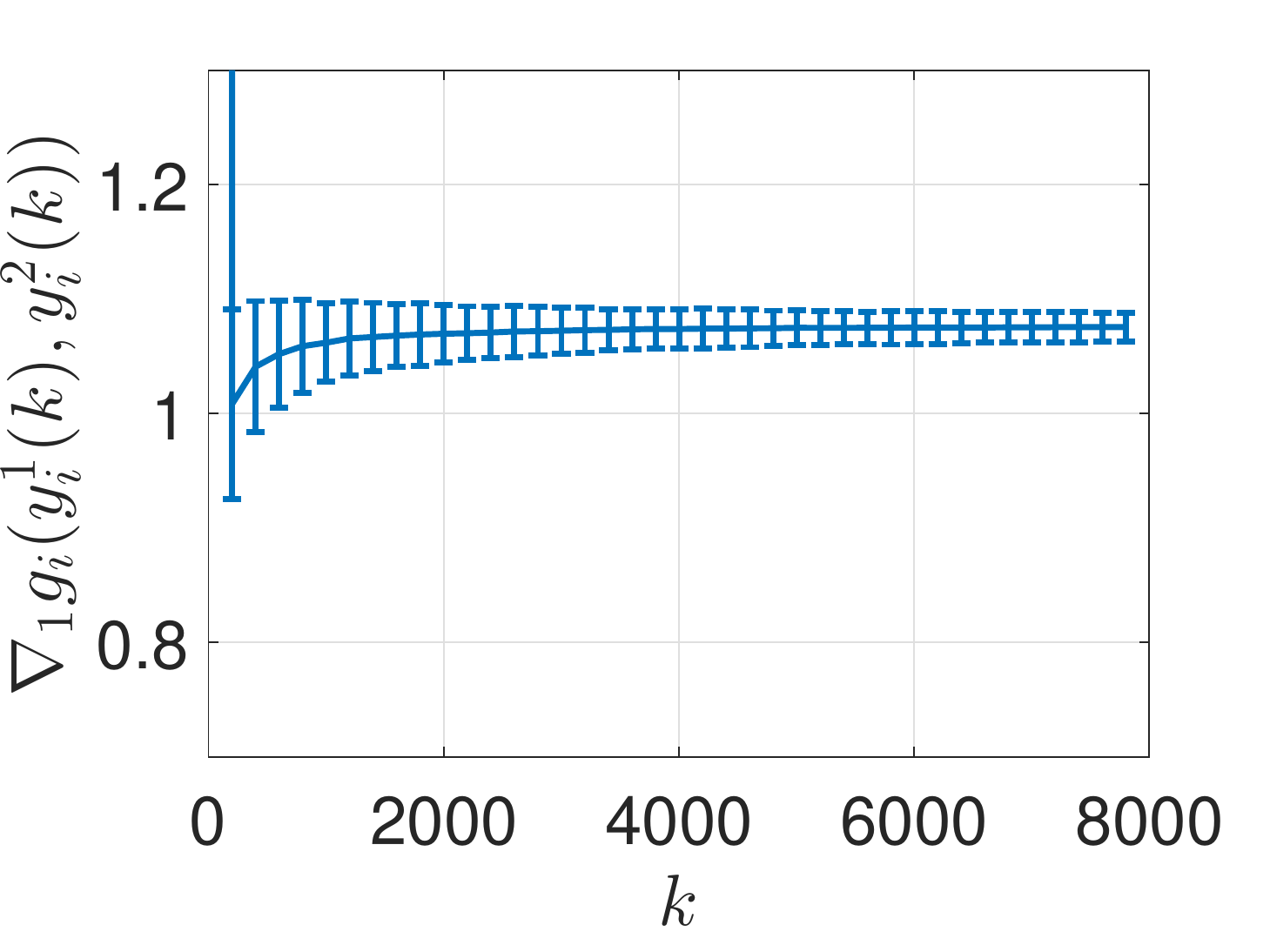}
			\captionof{figure}{Evolution of profile of derivatives \\of cost functions of all agents}
			\label{err_grad_BAIMD}
		\end{minipage}%
		\begin{minipage}{0.45\textwidth}
			\centering
			\includegraphics[width=1\linewidth]{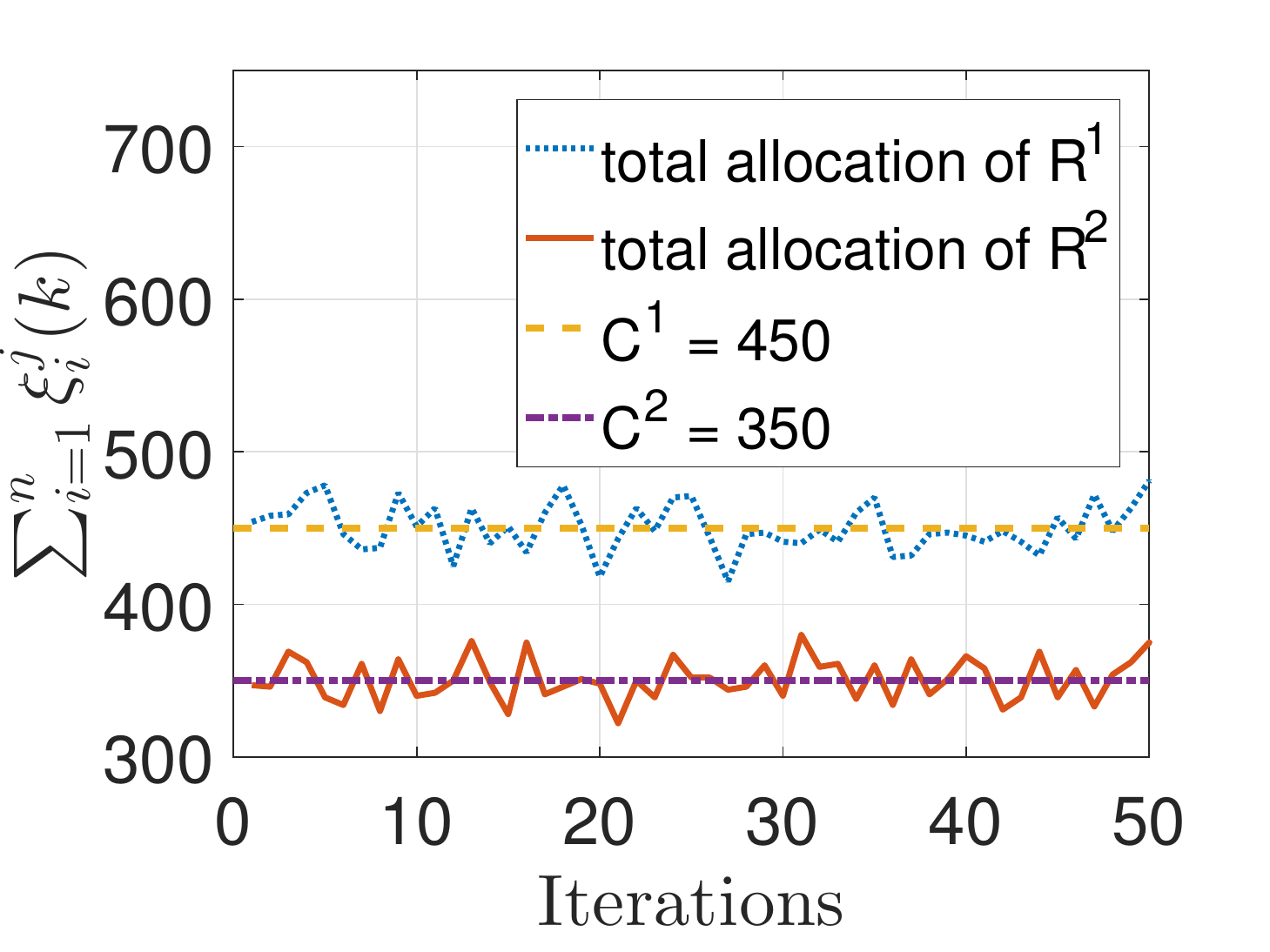}
			\captionof{figure}{Total allocation of resources over last $50$ time steps}
			\label{sum_alloc_BAIMD}
		\end{minipage}
	\end{figure}
    We see the total allocation of resources $R^1$ and $R^2$
	for last $50$ time steps in Figure (\ref{sum_alloc_BAIMD}). It is observed that most of the allocations are concentrated around their respective resource capacities. To reduce the overshoot of a resource $R^j$ we consider a constant $\gamma^j <1$ and modify the algorithm of control unit to calculate $\Omega^j(k+1)$ (cf. \eqref{omega}) in the following manner
	\begin{align*}
	\Omega^j(k+1) = \Omega^j(k) - \tau^j \left
	(\sum_{i=1}^n \xi^j_i(k) -\gamma^jC^j \right ),
	\end{align*}
	for all $j$ and $k$.
	
	\section{Conclusion} \label{conc} We proposed algorithms for solving the multi-variate optimization problems for capacity constraint applications in a distributed manner for divisible as well as indivisible unit-demand resources, extending a variant of AIMD algorithm. The features of the proposed algorithms are; it involves little communication overhead, there is no agent to agent communication needed and each agent has its own private cost functions. We observed that the long-term average allocation of resources reach the optimal values in both the settings. It is interesting to solve the following open problems: first is to provide a theoretical basis for the proof of convergence and second is to find the bounds for the rate of convergence, and its relationship with different parameters or the number of occurrence of capacity events. The work can also be extended in several application areas like smart grids, smart transportation or broadly in Internet of things where sensors have very limited processing power and battery life.
	
	\bibliographystyle{amsplain}
	\bibliography{DistOpt_bib}
\end{document}